\documentclass[11pt]{article}

\usepackage{arxiv}
\usepackage{amsmath}
\usepackage{amsthm}
\usepackage[utf8]{inputenc} 
\usepackage[T1]{fontenc}    
\usepackage{hyperref}       
\usepackage{url}            
\usepackage{booktabs}       
\usepackage{amsfonts}       
\usepackage{nicefrac}       
\usepackage{microtype}      
\usepackage{cleveref}       
\usepackage{graphicx}
\usepackage{doi}
\usepackage{enumitem} 
\usepackage{xcolor}

\usepackage{adjustbox}
\usepackage[ruled, vlined, linesnumbered, noresetcount]{algorithm2e}
\newtheorem{definition}{Definition}
\newtheorem{lemma}{Lemma}
\newtheorem{theorem}{Theorem}
\newtheorem{corollary}{Corollary}

\newcommand{\remove}[1]{}

\def \demand {demand}
\def \radius {radius around}

\def \G {{{G}}}
\def \H {{{H}}}

\def \t {r}

\def\tw{{\tt{tw}}}
\def\nd{{\tt{nd}}}
\def\cw{{\tt{cw}}}
\def\mw{{\tt{mw}}}
\def\dnd{{\tt{itp}}}
\def\vc{{\tt{vc}}}

\def \D {{\sc Domination}}

\def\H1{{H^{(1)}}}

\def\Hi1{{H^{(i-1)}}}

\def\mc{\texttt{mc}}
\def\K{\mathcal{K}}
\def\I{\mathcal{I}}
\def\G{\mathcal{G} }
\def\M{\mathcal{M} }

\newcommand{\Ni}[2]{N_{#1}(#2)}

\newcommand\xit[2]{x_{#1,#2}}

\title{Parameterized Complexity of $(d,r)$-Domination via Modular Decomposition}

\date{September 9, 2024}
\date{}

\newif\ifuniqueAffiliation

\ifuniqueAffiliation 
\author{ 
\href{https://orcid.org/0000-0000-0000-0000}{
\hspace{1mm}Gennaro Cordasco
\\
Department of Psychology, \\
University of Campania ''L.Vanvitelli'', Italy. \\ \texttt{gennaro.cordasco@unicampania.it}\\
	\And
	\href{https://orcid.org/0000-0000-0000-0000}
    Luisa Gargano, Adele A. Rescigno \\
	Department of Computer Science, \\
    University of Salerno, Italy.\\
	\texttt{\{lgargano,arescigno\}@unisa.it} \\

}
\else
\usepackage{authblk}

\setlength{\affilsep}{0em}
\newbox{\orcid}\sbox{\orcid}{\includegraphics[scale=0.06]{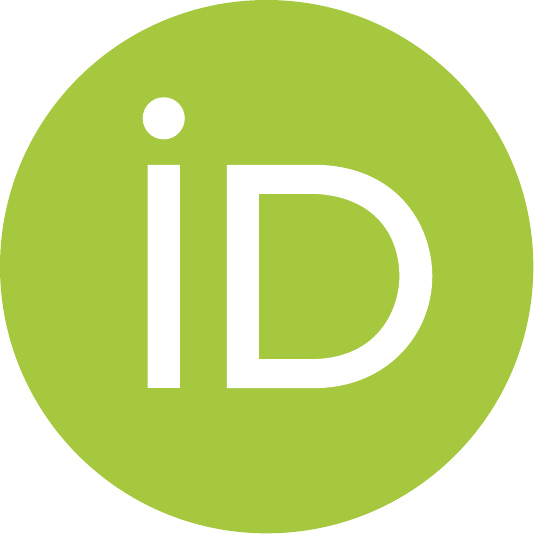}} 
\author[1]{%
	\href{https://orcid.org/0000-0001-9148-9769}
   {\usebox{\orcid}\hspace{1mm}Gennaro Cordasco
   }%
}
\author[2]{%
	\href{https://orcid.org/0000-0003-3459-1075}{\usebox{\orcid}\hspace{1mm}Luisa Gargano\thanks{\texttt{
Work  partially supported by project SERICS (PE00000014)
under the MUR National Recovery and Resilience Plan funded by the
European Union - NextGenerationEU.}}}%
}
\author[2]{%
	\href{https://orcid.org/0000-0001-9124-610X}{\usebox{\orcid}\hspace{1mm}Adele A. Rescigno$^*$}%
}
\affil[1]{Department of Psychology, 
University of Campania ''L.Vanvitelli'', Italy, {\tt{gennaro.cordasco@unicampania.it}}}
\affil[2]{Department of Computer Science, 
    University of Salerno, Italy, {\tt{lgargano@unisa.it}}}
    \affil[3]{Department of Computer Science, 
    University of Salerno, Italy, {\tt{arescigno@unisa.it}}}

\renewcommand{\shorttitle}{Parameterized Complexity of $(d,r)$-Domination}

\hypersetup{
pdftitle={Parameterized Complexity of (d,r)-Domination via Modular Decomposition},
pdfsubject={},
pdfauthor={Gennaro Cordasco, Luisa Gargano and Adele A. Rescigno},
pdfkeywords={Parameterized Complexity, Kernelization Algorithms, Domination, Contamination minimization},
}

\baselineskip=0.8truecm
\begin{document}
\maketitle

\begin{abstract}
	With the rise of social media, misinformation has significant
negative effects on the decision-making of individuals, organizations and
communities within society. Identifying and mitigating the spread of fake information is a challenging issue. We consider a generalization of the Domination problem that can be used to detect a set of individuals who, through an awareness process, can prevent the spreading of fake narratives. The considered problem, named \textsc{$(d,r)$-Domination} generalizes both distance and multiple domination. We study the parameterized complexity of the problem according to standard and structural parameters.  
We give fixed-parameter algorithms as well as polynomial compressions/kernelizations for some variants of the problem and parameter combinations. 
\end{abstract}

\baselineskip=0.7truecm

\section{Introduction} 
Domination is a fundamental concept in graph theory, which deals with the idea of dominating sets within a graph. In this problem, you seek to find the smallest set of vertices in a graph in such a way that every vertex in the graph is either in the dominating set or adjacent to a vertex in the dominating set. Dominating sets are critical in various real-world applications across fields that involve networks, connections, and coverage \cite{HHS1,HHS2}.
 
The Domination Problem in graph theory has several important variants that focus on different aspects of the problem. We focus on generalizations to distance domination and multiple domination that provide ways to solve practical problems related to the physical or operational constraints of networks. 
 
One application that inspires this paper is epidemiology. Graph-based information dissemination algorithms can be used to study the spread of real and fake information in a network. While these algorithms are not designed to intercept fakes, they can be used as part of a broader strategy to identify and mitigate the spread of fakes.
Controlling the spread of fake news is an ongoing challenge. Strategies to reduce the size of propagation by immunizing/removing vertices have been investigated in several papers \cite{barabasi,newman}.
In this paper, we focus on the effective production of accurate and evidence-based information to combat misinformation. To this end, we study the ($d,r$)-domination problem \cite{BHS}, which includes both multiplicity and distance.

{Apart from our inspiring application the ($d, r$)-domination problem, as the classical domination problem, has many applications. For instance it can be used, from the other side, as a mechanism to identify key influencers or individuals who can drive (mis)information diffusion within a network \cite{CCGMV,SNAM,CGMRV,CGRV,active,CGL+}. 
In wireless sensor networks, it helps minimize the number of sensors while ensuring full coverage and connectivity \cite{DB97}. In facility location and placement, it helps determine the optimal locations for services, facilities, or resources to ensure that they are accessible to a maximum number of people while minimizing the number of locations needed \cite{HHS98}.}

\paragraph*{The problem.}
Let $G=(V,E)$ be a graph. For a set of vertices  $X\subseteq V$,  we denote by $G[X]$  the induced subgraph of $G$ generated by $X$. 
Moreover, for  $v\in V$,  by\footnote{We will use the subscript $G$ whenever the graph is not clear from the context.} $\Ni{}{v}=\{u\in V \mid  (u,v)\in E\}$ the 
neighborhood of $v$.   
This notation can also be extended to  a set  $X\subseteq V$, we let
 $\Ni{}{X}=\{ y\in V\setminus X \mid  \exists x\in X, (x,y)\in E\}$ be the set of  vertices outside $X$  that have a neighbor in $X$.
For any two vertices  $u,v\in V$, we denote by $\delta(u,v)$ the distance between  $u$ and $v$ in $G$ 
and  by   $\Ni{\t}{v}=\{u \in V \mid u \neq v, \ \delta(u, v)\leq \t\}  $
the {\em neighborhood of radius $\t$ around $v$}.  Clearly, $\Ni{1}{v} = \Ni{}{v}$.    

A {\em Dominating set} in a graph $G = (V, E )$ is a subset of  $V$ such that every vertex not in the set has at least one neighbor in the set. 
In this paper, we will consider the following generalization  introduced in \cite{BHS}:
Given a {\em \demand} $d$ and a {\em \radius} $r$, a vertex $v \in V \setminus S$ is {\em dominated by $S$} if there exist at least $d$ vertices in $S$ that are at a distance at most $r$ from $v$.

In an information diffusion setting,
we consider a networked population of people who could be misled by a word-of-mouth dissemination approach. It can be assumed that an individual gets immune when exposed to enough debunking information. For further clarification, an individual is considered immunized when he/she receives the debunking information from a number of neighbors at least equal to its demand $d$.
Furthermore, every individual has a circle of trust (level of trust) defined by a radius $r$ around it. Debunking information from inside the circle of trust is the only source that can be trusted.

Specifically, we examine how to use $(d, r)$-dominating sets to identify a collection of individuals who can stop the spread of false narratives by initiating a debunking (or prebunking) immunization campaign. The vaccination operation on a vertex prevents the contamination of the vertex itself as well as the spread of such false narratives when there is a debunking  immunization campaign in place.
Therefore, to prevent the spread of malicious items, we are looking for a small subset S of vertices (also known as the immunizing set) which allows us to minimize the propagation of false information by disseminating the debunking information. Depending on the vertex demand $d$ and the radius $r$, which characterizes the vertices' circle of trust, the immunizing set should be able to cover each vertex $d$ times within a maximum distance of $r$.
\begin{definition}
Given an undirected graph $G=(V,E)$  and integers   $r,\, d>0$, a subset of vertices $S \subseteq V$ is  a  $(d,r)$-{\em dominating set} of $G$ if \\ $|N_r(v)\cap S|\geq d$,   for all $v \in V\setminus S$.
\end{definition}

In this paper, we study the $(d,r)$-\D\  problem that asks for a minimum $(d,r)$-dominating set.  
\begin{quote}
\textsc{$(d,r)$-Domination}:\\
\textbf{Input:} An  undirected graph $G{=}(V,E)$ and  two integers   $r,\, d{>}0$ \\   
\textbf{Output:} A $(d,r)$-dominating  set   of minimum size.
\end{quote}

{
Since the problem can be solved independently in each connected component of the input graph, from now on, we assume that the input graph is connected.}

\subsection{Parameterized Complexity and Compression/Kernelization}  
We study the $(d,r)$-\D\ problem from a parameterized point of view.
Parameterized complexity is a refinement of classical complexity theory in which one takes into account not
only the input size but also other aspects of the problem given by a parameter $p$.
A problem $Q$ with input size $n$ and parameter $p$ is called {\em fixed-parameter tractable (FPT)} if it can be solved in time $f(p) \cdot n^c$, where $f$ is a computable function only depending on $p$ and $c$ is a constant.
Downey and Fellows introduced a hierarchy of
parameterized complexity classes $FPT \subseteq W[1] \subseteq W[2] \subseteq \ldots$ that is believed to be strict, see \cite{DF99,FG}.

Once fixed-parameter tractability of a problem is established, we can
ask whether we can go even one step further by showing the existence of a polynomial compression/kernelization algorithm.

\def\PK{{PK}}
\def\PC{{PC}}

\begin{definition}
A compression 
for a parameterized problem $Q$ is a polynomial-time algorithm, which takes an instance  $(I,p)$ and produces an instance  $( I',p')$ called a kernel, such that    $(I,p)\in Q$ iff $(I',p') \in Q'$ for another problem $Q'$ 
and the size $|I'|+p'$  is bounded by $f(p)$ for some computable function $f$, called the size of the kernel. If $f$ is a polynomial function, the algorithm is called polynomial compression, denoted \PC. \\
If $Q=Q'$ then the algorithm is called kernelization. Moreover, if $Q=Q'$ and $f$ is a polynomial function then the algorithm is called polynomial kernelization, denoted \PK. 
\end{definition}

Kernelization can be seen as a special case of compression when the algorithm is required to 
produce another instance of the same problem.
{
Moreover by Theorem 1.6 in \cite{Fomin_Book2019}, we know that
if there is a polynomial compression of a problem $Q$ into another problem $Q'$ and the decision versions of $Q$ and $Q'$ are both NPC, then Q also admits a polynomial kernel.
The reason is that any two NPC problems have a polynomial-time reduction between them. So, for two NPC problems, polynomial compression and polynomial kernel are equivalent, if we do not consider kernel size. 
In this paper we will treat the two concepts separately in order to emphasize the fact that a kernelization is often a straightforward preprocessing strategy in which inputs to the algorithm are replaced by a smaller input, achieved by applying a set of reduction rules that cut away parts of the instance that are easy to handle.
}
It is known that for decidable problems, fixed-parameter tractability is equivalent to the existence of a kernelization algorithm, but not
necessarily a \PK.

\paragraph*{Related results.}

The  $(d,r)$-\D\ problem   was proved to be NP-complete even when restricted to bipartite graphs and chordal
graphs \cite{BHS}.
Bounds on the minimum size of a solution of $(d,r)$-\D\ are presented in \cite{JPSS}.
Subsequent results for some special classes of graphs are given in \cite{H,LHX}.
 $(d,r)$-\D\   generalizes several well-known, and widely studied, generalized domination problems.
 \begin{itemize} 
\item[-] When  $d=r=1$ the problem becomes the classic {\sc Dominating Set} problem.
\item[-] When $r=1$ and $d\geq 2$, the problem becomes the $d$-{\sc Domination} problem \cite{FKR23,Flink85}. For any $d$, this problem is known to remain NP-hard in the classes of split graphs.  From a positive point of view, it  is    solvable in linear time in the class of graphs having the property that   any block  is a clique, a cycle or a complete bipartite graph (this includes trees, block graphs, cacti, and block-cactus graphs),
and it is solvable in polynomial time in the class of chordal bipartite graphs  \cite{LC}.
The $d$-\D\ problem was also studied with respect to its (in)approximability properties, both in general \cite{CMV} and in restricted graph classes \cite{BFH}.
\item[-]When $d=1$ and $r\geq 2$, the problem becomes the $r$-{\sc Distance Domination} problem \cite{Henning20}. This problem, for any  $r$, remains NP-hard even if we restrict the graph to belong to certain special classes of graphs, including bipartite or chordal graphs \cite{CN}.
Slater presented a linear time algorithm for the $r$-{\sc Distance Domination} problem in forests \cite{S}.
An approximate algorithm for unit disk graphs was presented in \cite{DLDHR}.  
 \end{itemize}
From a parameterized complexity point of view, it is known that the $(d,r)$-\D\ problem, as well as the $d$-\D\ problem and the 
$r$-{\sc Distance Domination} problem are W[2]-hard with respect to the size $k$ of the solution since all these problems have the {\sc Dominating Set} problem as special case \cite{DF92}.
Moreover, since $(d,1)$-\D\ remains NP-hard for each $d\geq 2$ we have that $(d,r)$-\D\ is para-NP-hard with respect to $d$ and  W[2]-hard with respect to $d+k$, even when the radius $r$ is equal to $1$.

For graphs of bounded branchwidth, it  was presented in \cite{IOU16} an algorithm for the $d$-\D\ problem  
 parameterized by $d$ and the branchwidth {\tt bw}$(G)$ of graph $G$. Due to the linear relation between the branchwidth {\tt bw}$(G)$ and the treewidth $\tw(G)$ of graph $G$, the above result allows to obtain an algorithm for the $d$-\D\ problem  
 parameterized by $d$ and the treewidth $\tw(G)$. 
 An XP algorithm for the  $d$-\D\ problem parameterized by cliquewidth $\cw(G)$ of graph $G$ was presented in \cite{CCGMV}.
For the $r$-{\sc Distance Domination} problem, an algorithm parameterized by $r$ and the treewidth $\tw(G)$ was presented in \cite{BL16}.

In this paper, we are interested in
the analysis of  $(d,r)$-\D\  with respect to some structural parameters of the input graph $G$: The neighborhood diversity  $\nd(G)$, the iterated type partition $\dnd(G)$, and the modular-width $\mw(G)$ of  $G$.
FPT algorithms for   {\sc Dominating Set}  for  $\nd(G)$, $\dnd(G)$, and $\mw(G)$  parameters were presented in \cite{K}, \cite{DAM24}, and \cite{R}, respectively. 
Recently, Lafond and Luo \cite{LL} presented an FPT algorithm for  
$d$-\D\   
 parameterized by $\nd(G)$ and proved that this problem is W[1]-hard with respect to $\dnd(G)$.

The existence of \PK s for the  {\sc Dominating Set} and $r$-{\sc Distance Domination} problems with respect to the size $k$ of the solution in special classes of graphs has also been studied  \cite{DDFLPPVSSS16,ER,FLST}. No \PK\ for the {\sc Dominating Set} problem exists when parameterized by the solution size and the vertex cover number $\vc(G)$ \cite{DLS}. A  \PK\ for the {\sc Dominating Set} problem parameterized by $\dnd$, and consequently by $\nd$,
has been provided in \cite{DAM24}.

\paragraph*{Our results.}
In this paper, we give some 
results with respect to 
some standard and  structural parameters of 
$G$: the demand $d$, the modular-width $\mw(G)$, the iterated type partition number 
$\dnd(G)$ and the neighborhood diversity $\nd(G)$.

We present: 
\begin{itemize}
\item a FPT algorithm for $(d,r)$-\D\ 
parameterized by  
$\mw(G)+d$. 
\item a \PC\ for $(1,r)$-\D\ 
parameterized by $\mw(G)$  and, consequently, by $\dnd(G)$ and $\nd(G)$.
\item a \PC\ for $(d,r)$-\D\ parameterized by  $\dnd(G)+d$.
\item  a \PK\ for $(d,r)$-\D\ parameterized by $\nd(G)+d$.  
\end{itemize}
Table \ref{tableResults} gives a summary of known and new parameterized complexity results with respect to the considered parameters.

\begin{table}[tb!]
\begin{adjustbox}{max width=\textwidth}
\begin{tabular}{|l||c|c|c|c|}
\hline
Parameters                     & $(1,1)$-\D & $(1,r)$-\D  &  $(d,1)$-\D  & $(d,r)$-\D \\ 
\hline \hline
\nd              &  \PK\cite{DAM24}   &  {\bf \PC}[Th.\ref{1rDmw}]    & FPT\cite{LL}&   {\bf FPT}[Cor.\ref{cor9}] 
\\ \hline
\dnd            &  \PK\cite{DAM24}
& {\bf \PC}[Th.\ref{1rDmw}]  & W{[}1{]}-hard\cite{LL} &    W{[}1{]}-hard\cite{LL} 
\\ \hline
\mw             &  FPT\cite{R}   &  &W{[}1{]}-hard\cite{LL} &     W{[}1{]}-hard\cite{LL} \\ 
            &{\bf \PC}[Th.\ref{th7}]&{\bf \PC}[Th.\ref{1rDmw}]&  &      \\ \hline \hline
\nd\  and $d$     &   -    &         -     &  {\bf \PK}[Th.\ref{thnd}]&  {\bf \PK}[Th.\ref{thnd}]  \\ \hline
 \dnd\  and $d$     &   -    &         -     &  {\bf \PC}[Lem.\ref{th15}]&  {\bf \PC}[Th.\ref{cor16}]  \\ \hline 
 \mw\  and $d$     &   -    &         -     &  {\bf FPT }[Lem.\ref{th11}]&  {\bf FPT }[Th.\ref{cor12}]  \\ \hline 
\end{tabular}
\end{adjustbox}
\caption{Parameterized complexity 
with respect to  demand ($d$), neighborhood diversity (\nd),
iterated type partition number (\dnd) and
modular-width (\mw). \label{tableResults} 
} 
\end{table}

\section{Modular Decomposition} \label{mddef} 
The notion of modular decomposition of graphs was introduced by Gallai
 in \cite{Gallai}, as a tool to define hierarchical decompositions of graphs.

\begin{definition} \label{operationDef}
The  following operations can be used to construct a graph:
\begin{quote}
\begin{itemize}  
\item[(O1)] { The creation of an isolated vertex.}
\item[(O2)] {The disjoint union of two graphs denoted by $G_1\oplus G_2$, i.e., $G_1\oplus G_2$ is the graph with vertex set $V (G_1 ) \cup V ( G_2 )$ and edge set $E( G_1 )\cup E( G_2 )$. }
\item[(O3)] {The complete join of two graphs denoted by $G_1\otimes G_2$, i.e., $G_1\otimes G_2$ has vertex set $V (G_1 ) \cup V ( G_2 )$ and edge set 
$E( G_1 )\cup E( G_2 )\cup \{ (u, w) \mid  u \in V ( G_1 ),\  w\in  V(G_2) \}$.} 
\item[(O4)] 
The substitution of the vertices
$1,\ldots,\ell$ of an  outline graph $H$ by the graphs 
 $G_1,\ldots,G_\ell$, denoted by $H(G_1,\ldots,G_\ell)$, is the graph with vertices $\bigcup_{1\leq i\leq \ell} V(G_i)$ and edge set
$\bigcup_{{1\leq i\leq \ell}} E(G_i)\cup \{(u, w) \mid u \in G_i, \ w \in G_j, \  (i, j) \in E(H), \ 1 \leq i \neq j \leq \ell\}.$ 
\end{itemize}
\end{quote}
\end{definition} 
\noindent
Notice that (O2) and (O3) are special cases of (O4) with $\ell=2$. 

A {\em module} of a graph  $G = (V, E)$  is a set $M \subseteq V$ such that for all
$ u,v\in M, \  N(u)\setminus M=N(v)\setminus M.$
Two modules $M$ and $M'$ are adjacent if, in $G$,  every vertex of $M$ is adjacent to every vertex of $M'$, and independent if no vertex of $M$ is adjacent to a vertex of $M'$.
The empty set, the vertex set $V$, and all singletons $\{v\}$ for each $v\in V$ are always modules and they are called {\em trivial modules}. A graph is called {\em prime} if all of its modules are trivial.
A partition $\M=\{M_1,M_2,\ldots, M_\ell\}$ of $V$ is called a {\em modular partition} of $G$ if every element of $\M$ is a  module of $G$.
Hence, the modular partition of $G=G_1\oplus G_2$ and $G=G_1\otimes G_2$ in operations (O2)-(O3) is $\M=\{M_1,M_2\}$ where $G_1=G[M_1]$ and $G_2=G[M_2]$ and the modular partition of $G=H(G_1,\ldots,G_\ell)$ in operation (O4) is $\M=\{M_1,\ldots,M_\ell\}$ where for each\footnote{For a positive integer $a$, we use $[a]$ to denote the set of integers $[a] = \{1, 2, \ldots, a\}$.} $i\in [\ell], G_i=G[M_i]$.

\smallskip

Operations (O1)-(O4) taken to construct a graph, form a parse-tree of the graph. A \textit{parse tree} of a graph $G$ is a tree $T(G)$  that captures the decomposition of $G$ into modules.   The  leaves  of $T(G)$  represent  the  vertices  of $G$ (operation (O1)).  The internal vertices of  $T(G)$ capture operations on modules: Disjoint union of its children (operation (O2)), complete join (operation (O3)) and substitution (operation (O4)).     Figure \ref{fig:imgMW} depicts a graph $G$ and the corresponding parse tree.

\paragraph*{Modular-width.}
The Modular-width parameter was introduced in \cite{GLO}.
The \textit{width} of a graph $G$ is the maximum size of the vertex set of an outline graph $H$ used in
operation (O4) to construct $G$, if any, it is 2 otherwise; this number corresponds to the maximum number of children of a vertex in the corresponding parse tree. 
The \textit{modular-width} is the minimum width such that
$G$ can be obtained from some sequence of operations (O1)-(O4). Finding a parse-tree, of
a given graph $G$, corresponding to the modular-width, can be done in linear-time \cite{TCH+}.

\begin{figure}[tb!]
	\centering
	\includegraphics[width=0.97\textwidth]{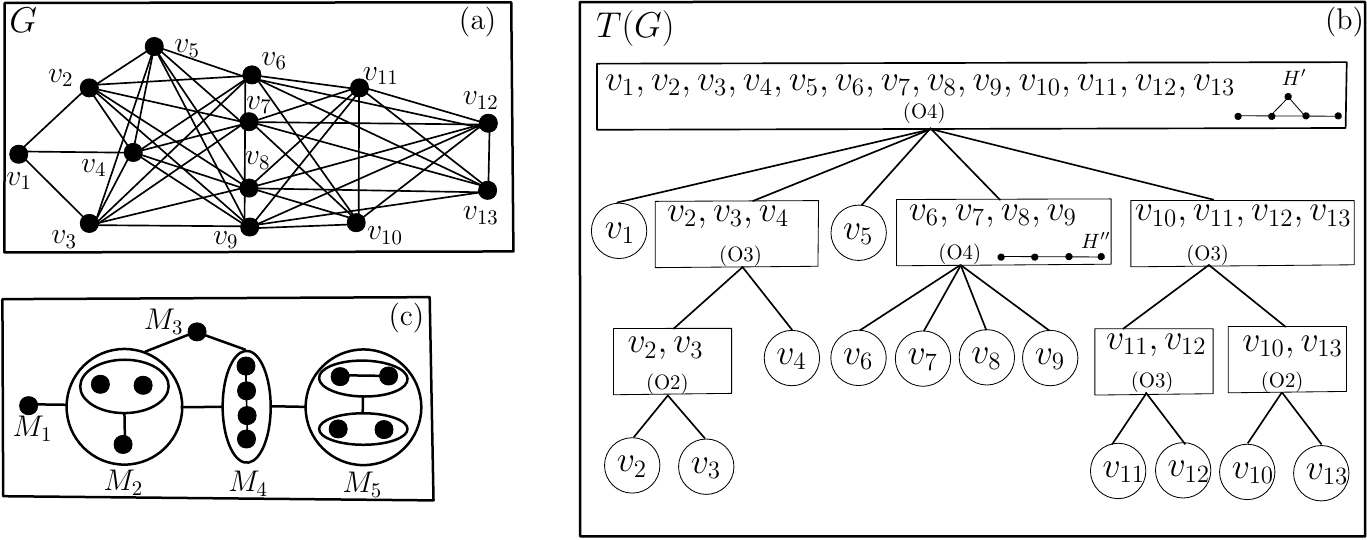}  
		\caption{(a) A graph $G$. (b) The parse tree $T(G)$ associated with a decomposition of $G$ into modules.
The width of the presented decomposition is $5$. 
  (c) A hierarchical representation of the decomposition of $G$ into modules.
  \label{fig:imgMW}}  
\end{figure}

\paragraph*{$\G$-modular cardinality.}
A variant of modular-width restricted to fixed graph classes has been proposed in \cite{LL}.
A graph class $\G$ is a (possibly infinite) set of graphs containing at least one non-empty graph.   
\begin{definition}\cite{LL}. Let $\G$ be a graph class. 
For a   graph $G$ (not necessarily in $\G$), a modular partition $\M = \{M_1,\ldots, M_\ell\}$ of  $G$ is called a $\G$-modular partition if each $G[M_i]$ belongs to $\G$.
The $\G$-modular cardinality of $G$, denoted by $\G$-$\mc(G)$, is the minimum cardinality of a
$\G$-modular partition of $G$.
\end{definition}  
Therefore, a graph $G$ of $\G$-modular cardinality  $\ell\geq 1$  and modules  $M_1,\ldots, M_\ell$, can be seen as $G=H(G_1,\ldots, G_\ell)$, according to substitution operation (O4) (including (O2) or (O3) when $\ell=2$), in which $G_i=G[M_i]\in \G$, for each $i\in [\ell]$. 
If $\ell=|V(G)|$, we say that $G$ is prime with respect to the class $\G$.

\paragraph*{Neighborhood diversity.}
The neighborhood diversity  $\nd(G)$ of a graph $G$ was introduced by Lampis in \cite{L}.
It can be seen as the $(\K \cup \I)${\em-modular cardinality} of $G$, where  $\K$  and   $\I$ denote the class of {\em clique} and of {\em independent set} graphs, respectively.  
Hence, a graph $G$ of neighborhood diversity $\ell\geq 1$  and modules  $M_1,\ldots, M_\ell$, can be seen as $G=H(G_1,\ldots, G_\ell)$, in which $G_i=G[M_i]\in \K \cup \I$, for each $i\in [\ell]$.

\paragraph*{Iterated type partition.}\label{sec-itp}
\def\C{\cal C}
Given a graph $G$, the {\em iterated type partition}  of  $G$, introduced in \cite{DAM24},
is defined by iteratively contracting {\em clique} and  {\em independent set} modules until no more contractions are possible; that is,
the obtained graph is prime with respect to the class $\K \cup \I$.
The iterated type partition number, denoted $\dnd(G)$, is the number of vertices of the obtained prime graph.
An example of a graph $G$ with $\dnd(G)=5$ and its iterative identification is given in Figure \ref{fig:img1}. 
\begin{figure}[tb!]
	\centering
	\includegraphics[width=0.97\textwidth]{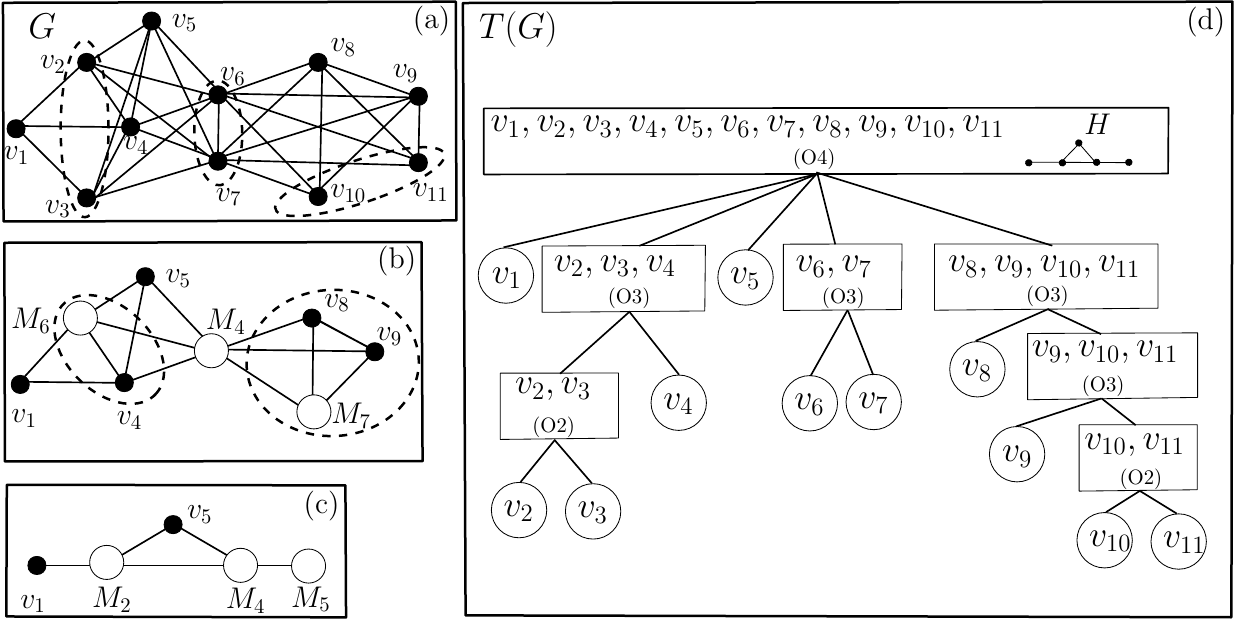}  
		\caption{(a)-(c) A  graph $G$ with {\em iterated type partition} number $5$ and its iterative identification. Dashed circles describe the identified {\em clique} or {\em independent set} modules. (d) The parse tree $T(G)$ associated with a decomposition of $G$ into modules. 
  Apart from the root, all internal vertices use only operations (O2) and (O3).
\label{fig:img1} 
} 
\end{figure}
It can be shown that the vertices of the obtained {\em prime graph} represent
modules that are {\em cographs}.
\remove{
Any cograph may be constructed using previously defined operations (O1), (O2), and (O3). Namely  the class   $\C$  of cographs can be recursively defined as:    
\begin{itemize}
\item any single-vertex graph is a cograph;
\item  $G_1\oplus G_2\in \C$, for any  $G_1, G_2\in \C$. 
\item   $G_1\otimes G_2\in \C$, for any  $G_1, G_2\in \C$.    
\end{itemize}

\noindent
}
Therefore, the {\em iterated type partition} number of a graph is equivalent to its $\C${\em-modular cardinality,} where $\C$ denotes the cograph class.

We know that for each graph $G$ we have $\mw(G) \leq \dnd(G) \leq   \nd(G)$. See also \cite{DAM24} for a discussion on the relations between modular decomposition parameters and other parameters.

{
\subsection{Modular Decomposition on Power Graphs }  
The {\em $r$-th graph power} of a graph $G=(V,E)$, denoted by $G^r$,
 is a graph  that  has the same set of vertices of $G$ and an edge exists between two vertices if and only if their distance in $G$ is at most $r$, that is, $V(G^r)=V$  and 
 $E(G^r)=\{(u,v)\ |\ u,v\in V,\ \delta_G(u,v)\leq r\}$.

\begin{lemma} \label{power} 
Let $G=(V, E)$ be any connected undirected graph. It holds
\begin{enumerate}[label=(\roman*)]   
 \item for every graph class $\G$ such that $\K\subseteq \G,$ $\G$-$\mc(G^r)\leq \G$-$\mc(G)$;
\item  $ \mw(G^r) \leq \mw(G).$
\end{enumerate}    
\end{lemma}   

Leveraging Lemma \ref{power},  we establish the following result.

\begin{theorem} \label{thGeneral}
For any $d,r>0$, if the $(d,1)$-\D\ problem is FPT  parameterized by  \mw\  (resp. $\G$-$\mc(G)$   with  $\K\subseteq \G$), then the $(d,r)$-\D\ problem is FPT  parameterized by  \mw\  (resp. $\G$-$\mc(G)$   with  $\K\subseteq \G$).
\end{theorem}

\begin{proof}{}
We know that there exists an algorithm {$\mathcal A_{\mw}$} (resp. {$\mathcal A_{\mc}$}) that solves the $(d,1)$-\D\ problem within time $f(\mw(G)) \cdot n^c$  (resp. $f(\G$-$\mc(G)) \cdot n^c$) where $c$ is a constant.  
Starting from $G$ one can easily build in polynomial time $G^r$. Then we observe that  
a  $(d,r)$-{\em dominating set} of $G$ corresponds to a $(d,1)$-{\em dominating set} of $G^r$ and vice-versa. 
Hence, using {$\mathcal A_{\mw}$} (resp. {$\mathcal A_{\mc}$}) on $G^r$ we are able to solve the $(d,r)$-\D\ problem within time $f(\mw(G^r)) \cdot n^{c'}$  (resp. $f(\G$-$\mc(G^r)) \cdot n^{c'}$) where $c'$ is a constant.   Moreover by Lemma \ref{power} we know that  $ \mw(G^r) \leq \mw(G)$ (resp. if $\K\subseteq \G,$ $\G$-$\mc(G^r)\leq \G$-$\mc(G)$) and the result holds. 
\end{proof}

Theorem \ref{thGeneral}  implies  that, by applying the FPT algorithm for {\sc Dominating Set}  parameterized by $\mw(G)$  (presented in \cite{R}), we immediately have that  $(1,r)$-\D\  is FPT parameterized by $\mw(G)$ and, consequently, by    $\dnd(G)$ and $\nd(G)$.
Moreover, recalling that both in case of $\nd(G)$ and $\dnd(G)$ modules are in $\K\cup \I$,  one can apply Theorem \ref{thGeneral} on the faster algorithms for {\sc Dominating Set}  parameterized by $\nd(G)$ (presented in \cite{K}) and $\dnd(G)$  (presented in \cite{DAM24}).

Moreover, by  Theorem \ref{thGeneral}, we have that, using the FPT algorithm \cite{LL} for the 
$(d,1)$-\D\ problem on $G^r,$ the following result holds.
  
\begin{corollary} \label{cor9}
$(d,r)$-\D\ problem is FPT parameterized by \nd.
\end{corollary}

\def\K{\mathcal{K}}
\def\I{\mathcal{I}}

\section[(d,r)-\D]{A FPT algorithm for $(d,r)$-\D\ parameterized by \mw\ plus the demand $d$}

Given a graph $G = (V, E)$, and let $T(G)$ be the parse-tree of $G$ where the number of children of each vertex in $T(G)$, representing the (O4) operation, is at most $\mw(G)$.  
Our strategy traverses the tree $T(G)$ in an opposite breadth-first fashion. 
For each vertex $w_M  \in  T(G),$ representing one of the operations (O1)-(O4) of Definition \ref{operationDef} and corresponding to a module $M$, in order to be able to reconstruct the solution recursively, we calculate optimal solutions under different hypotheses based on the following consideration.

Since all the vertices $v\in M$  share the same neighborhood $N(v)\setminus M$, we can solve the problem locally, considering that the solution, adopted by modules adjacent to $M$, dominates all the vertices in $M$ with the same multiplicity (i.e., for each $u,v \in M, \ |N(v)\cap (S\setminus M)|=|N(u)\cap (S\setminus M)|$).
Hence, we can partition the solution $S$ into two sets $S'\subseteq V\setminus M$ and $S''\subseteq M$. If we assume that $|N(v) \cap S'|=a,$ then all the vertices in $M$ are dominated $a$ times by vertices in $S'$ and we need to build a solution $S''\subseteq M,$ which dominates each vertex $v \in M$ at least $t=\max\{d-a,0\}$ times, that is we need a $(t,1)$-Dominating set of $G[M]$. 
Hence, we compute the optimal solution for $G[M]$ considering all the possible demands $t=0,1,\ldots,d$, where $t=0$ corresponds to the case $a\geq d$ (i.e., all the vertices are already dominated by at least $d$ neighbors) and, $t=d$ corresponds to the case $a=0$ (i.e., adjacent modules do not contain 
any vertex in $S'$).

{
\begin{definition} \label{ModulesCosts}
For each vertex $w_M \in T(G)$, representing a module $M$ and for each $t \in \{0,1,\ldots,d\}$ we define the cost of the module $M$ with demand $t$, denoted $c_M(t)$, as the size of a minimum $(t,1)$-Dominating set of $G[M].$
\end{definition}
  
It is worth observing that: 
\begin{itemize} 
\item For each module $M,$ we have $c_M(0)=0$.
\item For each module leaf $w_{\{v\}}$ of the parse tree $T(G)$, which corresponds to a vertex $v\in V$, we have $c_{\{v\}}(t)=1,$ for each $t \in [d].$
\item The size of the solution of our $(d,1)$-\D\ problem on $G$ corresponds to $c_V(d)$ (i.e., the cost of the root vertex for $t=d$).
\end{itemize}

Recalling that operations (O2) and (O3) of Definition \ref{operationDef} are special case of (O4), the following Algorithm \ref{algmwd} enables us to compute the costs of an internal vertex $w_M$ representing  the (O4) operation, that is $G[M]=H(G_1, \ldots, G_\ell),$ in which 
$\M=\{M_1,M_2,\ldots, M_\ell\}$ is a modular partition of  $G[M]$, $G_i=G[M_i]$ and $\ell \leq \mw(G),$ assuming that the costs for the children modules $G_1, G_2,\ldots, G_\ell$ have already been computed. 

\begin{algorithm}[tb!]
 \SetCommentSty{footnotesize}
  \SetKwInput{KwData}{Input}
 \SetKwInput{KwResult}{Output}
 \DontPrintSemicolon
\caption{{\sc ModuleCostComputation}
\label{algmwd}}
 \KwData {A graph $G[M]=H(G_1,\ldots, G_\ell)$, 
 costs $c_{M_i}(t),$ for each $i {\in} [\ell]$ and $t {\in} [d]$. }
\KwResult{The costs $c_M(t),$ for each  $t \in [d].$ }
\setcounter{AlgoLine}{0}
    \For{$t =1,2, \ldots,d$}{
    $c_M(t)=\infty$\\
    \For{{\bf each} $t_1,t_2,\ldots,t_\ell \mbox{ such that } 0\leq t_i \leq t, \forall i \in [\ell]$}{
       Find $x_1, x_2, \dots, x_\ell$ satisfying $\mbox{ILP}(G[M],t, (t_1,t_2,\ldots,t_\ell),(c_{M_1}(*),c_{M_2}(*),\ldots,c_{M_\ell}(*)))$
        $c_M(t)=\min\{ c_M(t), \sum_{i=1}^\ell c_{M_i}(t_i)+x_i \}$
        }
   }
    \Return $c_M(*)$   
\end{algorithm}

$ \qquad  \qquad \mbox{ILP}(G[M],t, (t_1,t_2,\ldots,t_\ell),(c_{M_1}(*),c_{M_2}(*),\ldots,c_{M_\ell}(*)))$ 
\begin{eqnarray*}
& &\quad \min \sum_{i=1}^\ell c_{M_i}(t_i)+x_i   \\  
& &\quad \mbox{\textbf{subject to:}}  \\
&(1)& \quad  t_i + \sum_{j \; : \; (i,j) \in E(H)} ( c_{M_j}(t_j)+x_j)  \geq t  \quad   \forall \ i \in [\ell]   \\
&(2)&  \quad  0 \leq x_i \leq |V(G_i)| -  c_{M_i}(t_i) \qquad\qquad   \forall  i \in [\ell] \\
\end{eqnarray*}
{
To evaluate the time to solve the above ILP, we use a well-known result, stated in \cite{FLMRS,Len}: Any
$\ell$-Variable Integer Linear Programming Feasibility 
can be solved in time $O(\ell^{2.5\ell+o(\ell)} \cdot L)$ where $L$ is the number of bits in the input.
\begin{itemize} 
\item[] \textbf{$\ell$-Variable Integer Linear Programming Feasibility}
\\
\textbf{Instance:}  A matrix $A \in \mathbb{Z}^{r\times \ell}$ and a vector $b \in \mathbb{Z}^r$.
\\
\textbf{Question:} Is there a vector $x \in \mathbb{Z}^\ell$  such that $Ax \geq b$?
\end{itemize}
Hence,  observing that the considered 
 ILP uses at most $ \ell $  variables and the coefficients are upper bounded by $|V(G)|$, we have that it can be solved within time $O(\ell^{2.5\ell+o(\ell)} \cdot \log |V(G)|)$.
 }
\begin{lemma} \label{costAlg}
Let $G[M]=H(G_1,\ldots, G_\ell),$ in which $G_i=G[M_i]$ and for each $i \in [\ell], t \in [d]$ let $c_{M_i}(t)$ be the size of a minimum $(t,1)$-Dominating set of $G_i$. The Algorithm \ref{algmwd} computes, in time $O(\ell^{2.5\ell+o(\ell)} d (d+1)^{\ell} \log|V(G)|)$, for each $t \in [d],$ all the values $c_M(t)$.
\end{lemma}

\begin{lemma} \label{th11}
$(d,1)$-\D\ is solvable in time $O(\mw^{2.5\mw+o(\mw)} d (d+1)^{\mw} |V(G)|\log|V(G)|)$.
\end{lemma} 
\begin{proof}{}
By observing that every vertex in the parse tree $T(G)$ has at least two children and there are exactly $|V(G)|$ leaves, we have that $T(G)$ has at most $|V(G)|-1$ internal vertices and the result in Lemma \ref{costAlg}. Hence, using the Algorithm \ref{algmwd}, for each internal vertex, we can easily build in time $O(\mw^{2.5\mw+o(\mw)} d (d+1)^{\mw} |V(G)|\log|V(G)|)$  the optimal cost $c_V(d)$ of the solution for the root vertex, which corresponds to the size of the solution of the problem. A standard backtracking technique can compute the optimal set $S$ within the same time.  
\end{proof}

By Lemma \ref{th11} and by using the same arguments of the proof of Theorem \ref{thGeneral}, the following result holds.   
\begin{theorem}\label{cor12}
$(d,r)$-\D\ is solvable in time $O(\mw^{2.5\mw+o(\mw)} d (d+1)^{\mw} |V(G)|\log|V(G)|)$. 
\end{theorem}

The next section shows a stronger result in case $d=1$.
}

\subsection[(1,r)-\D]{A \PC\ for $(1,r)$-\D\  parameterized by \mw}  

We shift our attention to the $(1,1)$-\D\ problem, that is the classical {\sc Dominating Set} problem on connected graphs. 

 Lafond and Luo \cite{LL2} proved the  nonexistence of a \PC\ for the {\sc Dominating Set} problem -- without the restriction of a connected input graph -- parameterized by $\mw$, while Luo \cite{Luo23} showed that the {\sc Dominating Set} problem -- without the restriction of a connected input graph -- has a polynomial Turing compression (PTC) parameterized by $\mw$. A PTC is a relaxed version of a \PC\ defined as follows.
\begin{definition}
A Turing compression 
for a parameterized problem $Q$ is a polynomial-time algorithm with the ability to access an oracle for another problem $Q'$, that can decide whether an instance  $(I,p) \in Q$ with queries, called Turing kernels, of size at most $f(p)$ for some computable function $f$. In particular, $Q$ admits a polynomial Turing compression (PTC) if $f$ is a polynomial function.
\end{definition}

In the following, we restrict our attention to connected graphs and prove the existence of a \PC\ for  $(1,1)$-\D\  parameterized by $\mw$, that is, 
if we constrain the {\sc Dominating Set} problem to connected graphs, then it admits a \PC\  parameterized by \mw.

Let $G=(V,E)$ be a connected undirected graph and let $T(G)$ be the parse-tree of $G$ where the number of children of each vertex in $T(G)$, representing the (O4) operation, is at most $\mw(G)$. 
Hence, there exists a modular partition $\M=\{M_1,M_2, \ldots, M_\ell\}$ where $\ell\leq \mw(G)$ and $G=H(G_1,\ldots, G_\ell),$ in which $G_i=G[M_i]$,  for some outline graph $H$ having $\ell$ vertices. 

\begin{lemma} \label{one}
Let $G=H(G_1,\ldots, G_\ell),$ in which $G_i=G[M_i]$, be a connected undirected graph.
There exists a solution $S$ of the  {\sc Dominating Set} problem for $G$ such that, for each $i \in [\ell].$ 
\begin{enumerate}[label=(\alph*)]
\item $|S \cap M_i| \leq 1$;
\item if $|S \cap M_i| =0$, then there exists $j \in \Ni{H}{i}$ such that $|S \cap M_j| = 1$;
\item if $|S \cap M_i| {=}1$ and for each  $j {\in} \Ni{H}{i}$ holds $|S \cap M_j| {=} 0$, then $\bigcap_{v \in M_i} (\Ni{G_i}{v} \cup\{v\}) {\neq} \emptyset$. 
\end{enumerate}
\end{lemma}

\noindent We now define a variation of  {\sc Dominating Set}  that will be useful to present the Polynomial Compression.

\begin{quote}
\textsc{Colored Domination}:\\
\textbf{Input:} A  graph $G=(V,E)$, 
a coloring function $c:V \rightarrow \{B, W\}$, and an integer $k$. \\   
\textbf{Output:} A  {\em colored dominating set}, that is, a set  $S \subseteq V$ such that \\
\hphantom{Out } $\bullet$  $|S| \leq k,$   \\ 
\hphantom{Out } $\bullet$  $N(v)\cap S\neq\emptyset$, 
for each $v \in V$ such that either $c(v)=B$  or $v\notin S$.
\end{quote}  

The following Theorem 
shows that  an instance $\langle G,k\rangle$  of the decision version of  {\sc Dominating Set}  on 
$G=H(G_1,\ldots, G_\ell)$ with parameter $\mw(G)$ can be reduced to an instance 
$\langle H, c, k\rangle$ 
of  {\sc Colored \D}. 

\begin{theorem} \label{th7}
  {\sc Dominating Set} admits a \PC\ with respect to \mw, for the class of connected graphs.
\end{theorem}  

By Lemma \ref{power}, we know that $ \mw(G^r) \leq \mw(G).$ Observing that a  $(1,r)$-{\em dominating set} of $G$ corresponds to a {\em dominating set} of $G^r$ and vice-versa, we have that the above strategy, applied on $G^r,$ allows to devise a \PC\ parameterized by \mw\ for the $(1,r)$-\D\ problem on $G$. By Theorem \ref{th7}, the following result holds.   
\begin{theorem} \label{1rDmw}
   $(1,r)$-\D\ admits a \PC\ with respect to \mw.
\end{theorem}

\section[(d,r)-\D]{A \PC\ for $(d,r)$-\D\ parameterized by \dnd\ plus the demand $d$} \label{dnd-d}  

This section is devoted to devising a \PC\ for $(d,1)$-\D\ problem parameterized by the iterated type partition number $\dnd$ plus the demand $d$.

Let $I=\langle G,k\rangle$ be an input of the decision version of the $(d,1)$-\D\ problem on $G=(V,E)$. 
As observed in Section \ref{mddef}, any graph $G=(V,E)$ is described by a modular partition $\{M_1,M_2, \ldots, M_\dnd\}$ where $\dnd=\dnd(G)$ and can be seen as $G=H(G_1,\ldots, G_\dnd)$,  for some outline graph $H$ having $\dnd$ vertices, and for each $i \in\{1,2,\ldots,\dnd\},$ $G_i=G[M_i]$ is a cograph.


 We are going to define an instance $I'$ of an Integer Linear Programming (ILP) problem on the variables $\xit{i}{t}\in \{0,1\}$, for each $i \in \{1,2,\ldots,\dnd\}$ and $t\in [d]$, which we will prove to be equivalent to $I$.

The ILP uses the costs $c_{M_i}(t)$  associated with the modules $M_i$, for each $i\in \{1,2,\ldots,\dnd\}$ and each $t \in [d]$, which can be obtained exploiting Algorithm \ref{algmwd}.
The binary variable 
$\xit{i}{t}=1$ means that we are choosing a solution $S_i\subseteq M_i$ such that $S_i$ is a minimum $(t,1)$-Dominating set of $G_i.$ By Lemma \ref{costAlg} we know that $|S_i|=c_{M_i}(t).$
The constraints of the ILP are as follows:  
\begin{eqnarray*}    
&(1)&\sum_{{ i\in \{1,2,\ldots,\dnd\}, \  t\in [d]}} c_{M_i}(t) \cdot \xit{i}{t} \leq k\\
&(2)&\sum_{{ (i,j) \in E(H), \  t\in [d]}} c_{M_j}(t) \cdot \xit{j}{t}  + \sum_{t\in [d]} t \cdot \xit{i}{t}\geq d  \qquad \forall i \in \{1,2,\ldots,\dnd\} \\
&(3)&\sum_{t \in [d]} \xit{i}{t} \leq 1\qquad \qquad \qquad \qquad  \qquad \qquad \qquad \qquad  \forall i \in \{1,2,\ldots,\dnd\}\\
&(4)& \xit{i}{t} \in\{0,1\}  \qquad \qquad \qquad \qquad \qquad \qquad \qquad \qquad   \forall i \in \{1,2,\ldots,\dnd\}, \forall t\in [d],   
\end{eqnarray*} 
where: 
\begin{itemize}    
\item the inequality (1) guarantees that the size of the considered solution is at most $k$;
\item the inequalities (2) guarantee that the vertices of each module $M_i$  are dominated $d$ times. Specifically, the first sum corresponds to the number of times a vertex in $M_i$ is dominated by vertices of the solution belonging to adjacent modules,  and the second sum corresponds to the number of times a vertex in $M_i$ is dominated by vertices of the solution belonging to $M_i$; 
\item the inequalities (3) guarantee that, for each module $M_i$ at most one of the minimum $(t,1)$-dominating set of $G_i$, for $t=1, \ldots d$,  is chosen (i.e., at most one among $c_{M_i}(1), c_{M_i}(2), \ldots, c_{M_i}(d$) is chosen).
 We notice that when $\sum_{t \in [d]} \xit{i}{t}=0,$ then no vertex in the solution is a vertex in $M_i$ and then, all the vertices in $M_i$ are dominated by vertices of the solution belonging to adjacent modules (that are by (2) at least $d$). 
\end{itemize}

 To obtain the desired \PC\ algorithm, we are going to use the Frank and Tardos' algorithm \cite{FT87} that enables to compress linear inequalities to an encoding length that is polynomial in the number of variables.  
\begin{theorem}{\cite{FT87}} \label{thFT}
There is an algorithm that, given a vector $w \in  \mathbb{Q}^r$ and an integer $N$, in polynomial time finds a vector $\overline{w} \in \mathbb{Z}^r$ with $\lVert \overline{w} \rVert_\infty \leq 2^{4r^3} N^{r(r+2)}$ such that $sign(w \cdot b) = sign(\overline{w} \cdot b)$ for all vectors $b \in \mathbb{Z}^r$ with $\lVert b \rVert_1 \leq N- 1$.
\end{theorem}
In particular, we will use the same approach as adopted in \cite{EK17} (see Cor. \ref{cor2}). 
\begin{corollary}{\cite{EK17}} \label{cor2}
 There is an algorithm that, given a vector $w \in  \mathbb{Q}^r$ and a rational $W \in \mathbb{Q}$ , in polynomial time finds a vector $\overline{w} \in \mathbb{Z}^r$ with $\lVert \overline{w} \rVert_\infty= 2^{O(r^3)}$ and an integer $\overline{W} \in \mathbb{Z}$ with total encoding length $O(r^4)$, such that $w \cdot x = W$ if and only if  $\overline{w} \cdot x = \overline{W}$  for every vector $x \in \{0,1\}^r$.
 \end{corollary}

The following Lemma 
shows that $I=\langle G,k\rangle$ and the ILP $I'$ are equivalent, that is, $G$ admits a $(d,1)$-Dominating set of size at most $k$ if and only if $I'$ admits a feasible solution. Moreover, we can compress $I'$ to a new equivalent instance $I''$ such that $|I''|$ is bounded by $f(\dnd,d)$, where $f$ is a polynomial function.

\begin{lemma} \label{th15}
The $(d,1)$-\D\ problem, parameterized by  \dnd\ plus $d,$ admits a \PC.
\end{lemma}   

By Lemma \ref{power}, we know that $ \dnd(G^r) \leq \dnd(G).$ Observing that a  $(d,r)$-{\em dominating set} of $G$ corresponds to a $(d,1)$-{\em dominating set} of $G^r$ and vice-versa, we have that the above strategy, applied on $G^r,$ allows to obtain a \PC\ parameterized by $\dnd$ plus $d$ for the $(d,r)$-\D\ problem on $G$. 
By Lemma \ref{th15},  the following result holds.  

\begin{theorem} \label{cor16}
The $(d,r)$-\D\ problem, parameterized by \dnd\ plus $d,$ admits a \PC.
\end{theorem}

\section[(d,r)-\D]{A \PK\ for $(d,r)$-\D\ parameterized by \nd\ plus the demand $d$} 

Since $\nd(G)\geq \dnd(G)$, Theorem \ref{cor16} applies also to neighborhood diversity. In the following, we show that in the latter case, we can go further by devising a \PK\ using a straightforward pruning strategy. 


Let $I=\langle G,k\rangle$ be an instance of the decision version of the $(d,1)$-\D\ problem on $G=(V,E)$, where $G=(V,E)$ is described by a modular partition $\{M_1,M_2, \ldots, M_\nd\}$ where $\nd=\nd(G)$ and can be seen as $G=H(G_1,\ldots, G_\nd)$,  for some outline graph $H$ having $\nd$ vertices, and for each $i \in \{1,2,\ldots,\nd\},$ $G_i=G[M_i]$ is a clique or an independent set.

The following Lemma shows that there exists an optimal solution $S$ such that for each module, the number of vertices belonging to the solution is bounded to $2d-1$.


\begin{lemma} \label{ndbound}
Let $G=H(G_1,\ldots, G_\nd),$ in which $G_i=G[M_i]$ is a clique or an independent set.
There exists a solution $S$ of the  $(d,r)$-\D\ problem for $G$ such that, for each $i \in \{1,2,\ldots,\nd\},$ $|S_i=S\cap M_i|\leq 2d-1.$
\end{lemma}


\begin{theorem} \label{thnd}
The $(d,r)$-\D\ problem, parameterized by  \nd\ plus $d,$ admits a \PK.
\end{theorem}   
\begin{proof}{}
Let $G=(V,E)$ be a graph.
We will reduce an instance $I=\langle G,k\rangle$ of the decision version of the $(d,1)$-\D\ problem on $G=(V,E)$,
to an instance $I'=\langle G',k\rangle$  of the same problem on a pruned graph $G'.$
We recall that $G=(V,E)$ is described by a modular partition $\{M_1,M_2, \ldots, M_\nd\}$ where $\nd=\nd(G)$ and can be seen as $G=H(G_1,\ldots, G_\nd)$,  for some outline graph $H$ having $\nd$ vertices. Specifically, $G'$ is obtained from $G$ pruning each module $M_i$, having more than $2d$ vertices, to any subset of $M_i$ having $2d$ vertices. 

We prove that $G$ admits a $(d,r)$-{\em dominating set} of size at most $k$ if and only if $G'$ admits a $(d,r)$-{\em dominating set} of size at most $k$.

First, assume that $S$ is a solution of the $(d,r)$-\D\ problem on $G$ and $|S|\leq k$.  By Lemma \ref{ndbound}, we may assume that for each $i \in \{1,2,\ldots,\nd\},$ $|S_i=S\cap M_i|\leq 2d-1.$

Let $S'= \bigcup_{i=1}^{\nd} S'_i$, where $S'_i$ is any set of $|S_i|$ vertices in $M_i.$  Clearly $|S'|=|S|.$
 Since $S$ is a solution of the $(d,r)$-\D\ problem $G$, then any vertex in $M_i$ is dominated by some vertices in $S_i$ (if $G_i$ is a clique) and by some vertices in $S_j$  with $j\neq i$. In both cases, the choice of nodes in subset $S_i$, for each $i \in \{1,2,\ldots,\nd\}$ does not matter. The only thing that matters is the cardinality of such sets. Since $S'$ keeps these cardinalities unchanged we have that
 $S'$ is a solution of the $(d,r)$-\D\ problem on $G'$.

Assume now that $S'$ is a solution of the $(d,r)$-\D\ problem on $G'$. Let $S=S'$, since all the nodes in a module share the same neighborhood then it is not hard to verify that $S$ is a solution of the $(d,r)$-\D\ problem on $G.$

Finally, we note that the number of vertices of $G'$ is at most $2d \times  \nd(G)$, hence $|I'|$ is bounded by $f(\nd,d).$ 
\end{proof}

\end{document}